\documentclass[10pt]{IEEEtran}

\usepackage{verbatim}
\usepackage{amsfonts,amsmath,mathrsfs,amssymb,amsbsy}
\usepackage[final]{graphicx}
\usepackage{times,cite}
\usepackage{balance}
\usepackage{caption}
\usepackage{subcaption}

\usepackage{enumitem,kantlipsum}

\long\def\symbolfootnote[#1]#2{\begingroup%
\def\thefootnote{\fnsymbol{footnote}}\footnote[#1]{#2}\endgroup}

\usepackage{verbatim}
\usepackage[]{float,latexsym}
\usepackage{amsfonts,amsmath,mathrsfs,amssymb,amsbsy}
\usepackage{url}
\usepackage{amsthm}

\newtheorem{theorem}{Theorem}

\newtheorem{lemma}{Lemma}
\newtheorem{definition}{Definition}

\newcommand{\Prob}{\mathsf{P}}
\newcommand{\Expect}{\mathsf{E}}

\usepackage{color}
\definecolor{lightblue}{rgb}{.7, .8, 1}
\definecolor{lightgreen}{rgb}{.6, 1, .6}
\usepackage{color}
\definecolor{brown}{rgb}{1,0.38,0.03}

\definecolor{OliveGreen}{rgb}{.2,0.6,0.2}
\definecolor{BrickRed}{rgb}{.7,0.2,0.2}

\newcommand{\ignore}[1]{} 

\long\def\symbolfootnote[#1]#2{\begingroup%
\def\thefootnote{\fnsymbol{footnote}}\footnote[#1]{#2}\endgroup}

\DeclareMathOperator*{\esssup}{ess\,sup}

\newcommand{\bsp}{\begin{split}}
\newcommand{\esp}{\end{split}}

\begin{document}

\sloppy

\title{Minimax-Optimal Algorithms for Detecting Changes in Statistically Periodic Random Processes}

\author{Taposh Banerjee, \textit{Member, IEEE},  \vspace{-3ex} Prudhvi Gurram, and Gene Whipps
}

\maketitle

\symbolfootnote[0]{\small
Taposh Banerjee is with the Department of Electrical and Computer Engineering, University of Texas at San Antonio, San Antonio, TX. 
Prudhvi Gurram is with Army Research Laboratory, Adelphi, MD, and with Booz Allen Hamilton, McLean, VA. 
Gene Whipps is with the Army Research Laboratory, Adelphi, MD. 
(e-mail: taposh.banerjee@utsa.edu; gene.t.whipps.civ@mail.mil; gurram\_prudhvi@bah.com)

The work of Taposh Banerjee was partially supported
by a grant from the U.S. Army Research Laboratory, W911NF1820295, and by the Cloud Technology Endowed Professorship IV. 
A preliminary version of this paper has been presented at IEEE ICASSP 2019.

}

\begin{abstract}
Theory and algorithms are developed for detecting changes in the distribution of statistically periodic random processes. 
The statistical periodicity is modeled using independent and periodically identically distributed processes, a new class 
of stochastic processes proposed by us. 
An algorithm is developed that is minimax asymptotically optimal as the false alarm rate goes to zero. 
Algorithms are also developed for the cases when the post-change distribution is not known or when there are multiple streams of observations. 
The modeling is inspired by real datasets encountered in cyber-physical systems, biology, and medicine. 
The developed algorithms are applied to sequences of Instagram counts collected around a 5K run in New York City to detect the run.  
\end{abstract}


\section{Introduction}
The problem of detecting an abrupt change in the statistical properties of a measurement process has many applications in engineering and sciences including statistical process control \cite{shew-jamstaa-1925, shew-book-1931, taga-jqt-1998, stou-etal-jamstaa-2000}, sensor networks \cite{veer-ieeetit-2001, tart-veer-fusion-2002, tart-veer-fusion-2003, mei-ieeetit-2005, bane-etal-ieeetwc-2011, bane-tsp-2015}, computer networks \cite{tart-et-al-stamethod-2006, baras-etal-ieeenet-2009}, public health \cite{frisen-sqa-2009, fien-shmu-statmed-2005, baro-bookchapter-2002}, and telecommunication engineering 
\cite{arun-vinod-icumt-2009, arun-vinod-ukiwcws-2009}. In many of these applications, algorithms are needed to detect the changes in real time, that is the 
changes have to detected as soon as they occur. Mathematical problems for real-time detection are studied 
in an area of statistics called sequential analysis \cite{wald_book_1947, sieg-seq-anal-book-1985, wood-nonlin-ren-th-book-1982}. Specifically, in the problem of quickest 
change detection (QCD), theory and algorithms are developed to detect a change in the distribution of a sequence of random variables with the minimum possible delay subject to a contraint on the rate of false alarms \cite{page-biometrica-1954, lord-amstat-1971, poll-astat-1985, mous-astat-1986, lai-ieeetit-1998, tart-veer-siamtpa-2005, tartakovsky2017asymptotic, Pergamenchtchikov2018}. See \cite{veer-bane-elsevierbook-2013, poor-hadj-qcd-book-2009, tart-niki-bass-2014} for surveys. 

In this paper, we study the problem of change detection for applications where the observation process exhibits statistically periodic behavior. 
Below we give three such applications:
\begin{figure}
\centering
\includegraphics[scale=0.35]{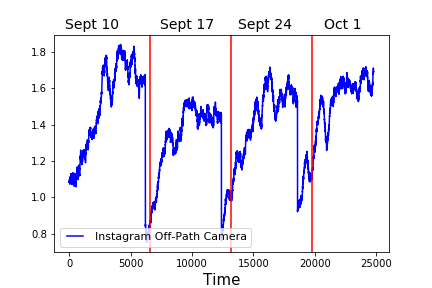}
\hspace{-0.5cm}
\includegraphics[scale=0.35]{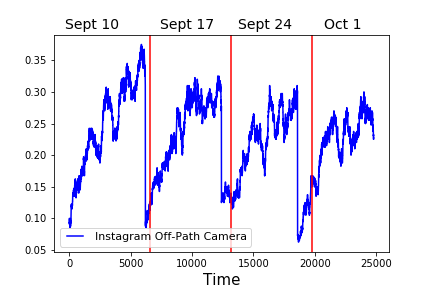}
\includegraphics[scale=0.35]{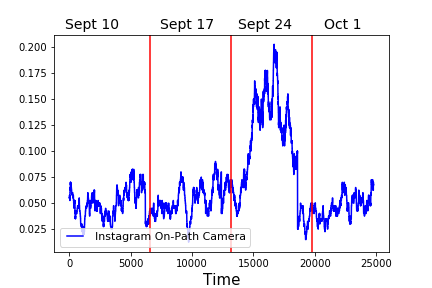}
\caption{Cyber-physical system application: The average Instagram post counts for data collected in NYC in \cite{bane-fusion-2018}. The first two figures show that the average counts 
have similar statistical properties across different days. The third figure shows the average Instagram count data for a CCTV camera on the path of the event and shows an increase in the average counts on the event day, Sept. 24.}
\label{fig:personInsta}\vspace{-0.4cm}
\end{figure}
\begin{enumerate}[leftmargin=*]
\item Cyber-physical systems: In Fig.~\ref{fig:personInsta}, we have plotted the average number of Instagram messages posted near three CCTV cameras in New York City (NYC). These counts were extracted from multi-modal traffic data (CCTV images, Twitter and Instagram messages) that we collected in NYC \cite{bane-fusion-2018}, \cite{bane-globalsip-2018}. We collected data during a 5K run that occurred on September 24, 2017. We also collected data on two Sundays before and one Sunday after the event day. The first two plots in 
Fig.~\ref{fig:personInsta} are for two CCTV cameras that were not on the path of the 5K run while the third figure contains data corresponding to a camera on the path of the run. As can be observed from the figures, the average message counts show similar statistical behavior on the four Sundays for the off-path cameras. For the on-path camera, the average message count increases due to the event. This periodic pattern was also observed in other data that we collected in  \cite{bane-fusion-2018} and \cite{bane-globalsip-2018}. Thus, the problem of anomaly detection or change detection in traffic data can be posed as the problem of detecting deviations away from statistically periodic behavior. 
\item Neuroscience: In Fig.\ref{fig:NeuralData}, we have plotted the neural spike data (binned) collected from a brain-computer interface (BCI) study on mice \cite{zhang2003two}. In this BCI experiment, it is of interest to investigate if an observer mouse learns to associate a cue to a shock given to a target mouse. In the first $15$ trials (below dashed horizontal line), no shocks are given. Starting trial number $16$, a cue is followed by a shock. The time of the cue is indicated by the dashed vertical line in the figure. The change point problem of interest here is to detect a change in the neural firing patterns caused by the behavioral learning starting trial number $16$. If a change is detected, then further shock trials can be avoided. The neural spike data shows statistical periodicity across trials as the observer mouse has to be part of the same experiment in every trial. For further discussion, we refer the readers to \cite{bane-NER-2019}. 
\begin{figure}
\centering
\includegraphics[scale=0.3]{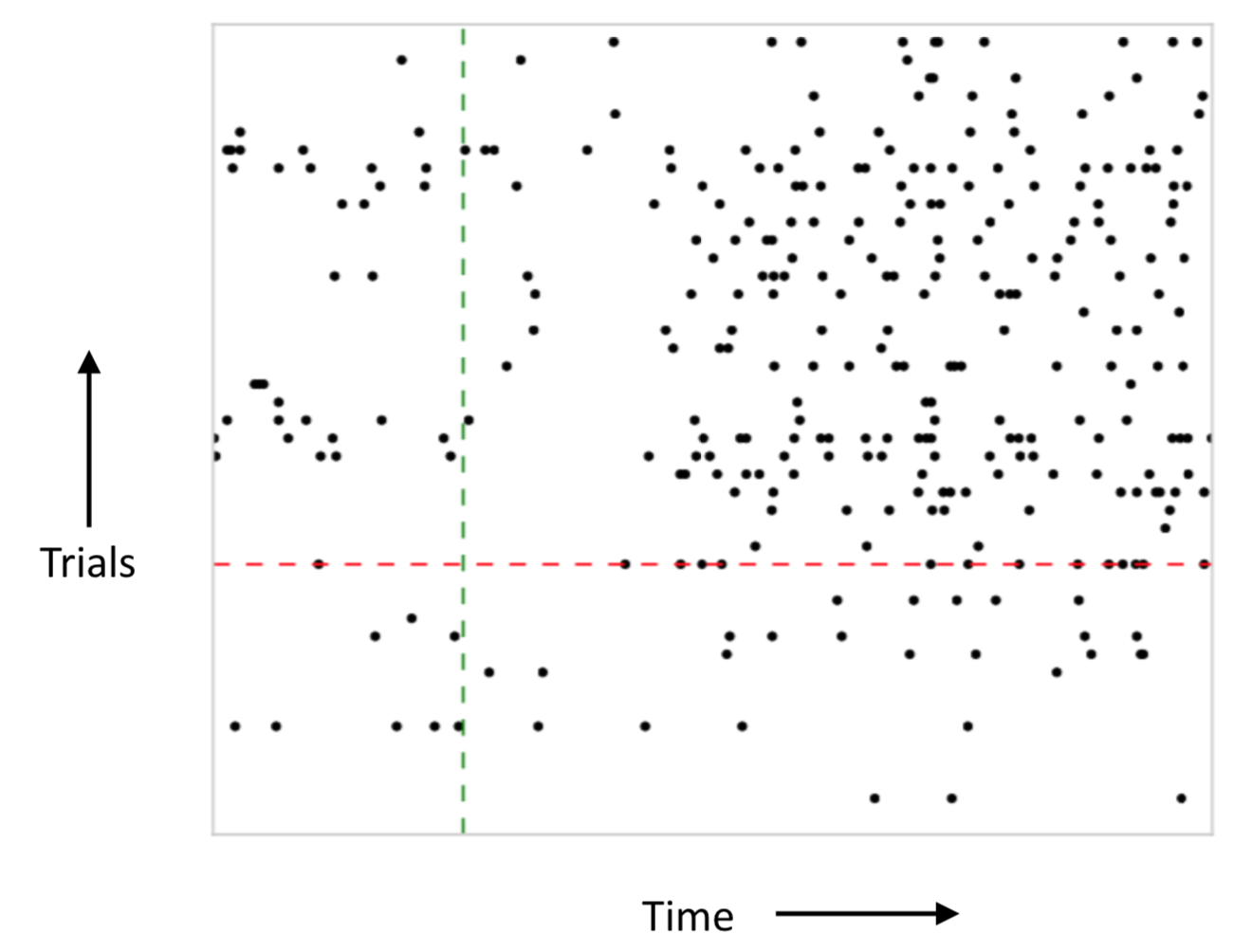}
\caption{Neuroscience Application: Figure taken from \cite{zhang2003two} of neural firing data collected from an observer mice in a BCI experiment. }
\label{fig:NeuralData}
\end{figure}
\item Medicine: Signals collected in electrocardiography (ECG) show regular patterns of P waves, QRS complexes, and ST segments; see Fig.~\ref{fig:ECG}. It is well-known that  arrhythmia 
can change the shape of this pattern. Typically, a 12-lead ECG is used to record a patient's ECG data and a human expert reads a chart to detect the anomaly. This process can be automated and can be done in real time using a change-point detection algorithm. 
\begin{figure}
\centering
\includegraphics[scale=0.25]{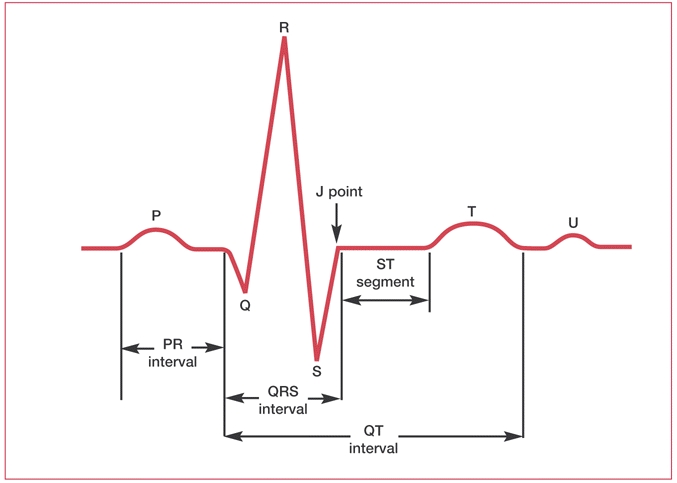}
\includegraphics[scale=0.25]{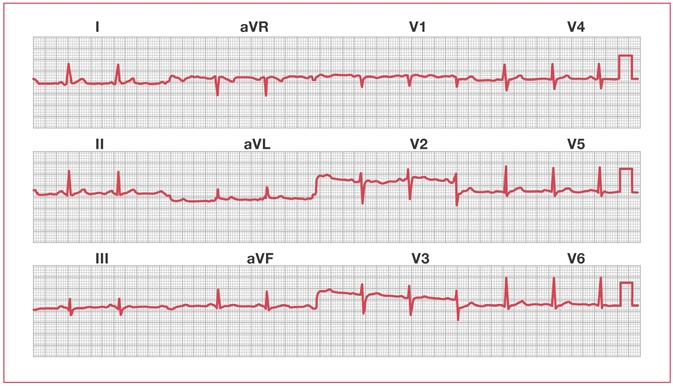}
\caption{Biology Application: Figure taken from \cite{ashley2004cardiology} shows normal ECG signal. The ECG signal is an approximately periodic sequence of waves containing a P wave, a QRS complex and an ST segment. }
\label{fig:ECG}
\end{figure}
\end{enumerate}

\medskip

The classical QCD literature can be broadly classified into two categories: results for independent and identically distributed (i.i.d.) processes with algorithms that can be computed recursively and enjoy strongly optimality properties \cite{mous-astat-1986}, and results for non-i.i.d. data with algorithms that are hard to compute but are asymptotically optimal \cite{lai-ieeetit-1998, tart-veer-siamtpa-2005, tartakovsky2017asymptotic, Pergamenchtchikov2018}. As can be seen in the figures (Fig.~\ref{fig:personInsta}--Fig.~\ref{fig:ECG}), the data encountered 
in applications of interest cannot be modeled as i.i.d. processes.  
To model statistically periodic data, such as shown in these figures, we use a new class of stochastic processes called independent and periodically identically distributed processes (i.p.i.d.). This is a class of processes introduced by us in \cite{bane-ieeeit-2019}. For this class of non-i.i.d. processes, we will show in this paper that the optimal algorithms can be computed recursively and using finite memory. This makes the 
developed algorithms amenable to implementation in various applications in cyber-physical systems, medicine, and biology. 

In \cite{bane-ieeeit-2019}, we provided a Bayesian optimality theory of change detection in i.p.i.d. processes where it was assumed that a prior distribution of the change point variable is available to the decision maker. Since such an assumption is typically not satisfied in practice, in this paper we take a non-Bayesian approach. Specifically, we obtain algorithms that are asymptotically optimal for two minimax stochastic optimization formulations \cite{lord-amstat-1971}, \cite{poll-astat-1985}. We also develop algorithms when the post-change distribution is unknown or when there are multiple streams of observations. 
Finally, we apply the developed algorithms to the NYC traffic data to detect the 5K run.

\section{Model and Problem Formulation}\label{sec:modelproblem}
In this section, we discuss the concept of an i.p.i.d. process which we will use to model statistically periodic data. We also discuss a change point model and 
two minimax problem formulations. 

\subsection{Statistical Model}
An i.i.d. process is a sequence of random variables that 
are independent and have the same distribution. To model the periodic statistical behavior of data, we use a new class of stochastic processes called 
i.p.i.d. processes which are defined as follows. 

\begin{definition}
Let $\{X_n\}$ be a sequence of random variables such that the variable $X_n$ has density $f_n$. The stochastic process $\{X_n\}$
is called independent and periodically identically distributed (i.p.i.d) if $X_n$ are independent and there is a positive integer $T$ such 
that the sequence of densities $\{f_n\}$ is periodic with period $T$:
$$
f_{n+T} = f_n, \quad \forall n \geq 1. 
$$
We say that the process is i.p.i.d. with the law $(f_1, \cdots, f_T)$. 
\end{definition}
If $T=1$, then the i.p.i.d. process is an i.i.d. process. 
The law of an i.p.i.d. process is completely characterized by the finite-dimensional product distribution involving $(f_1, \cdots, f_T)$. 
We assume that in a normal regime, the data can be modeled as an i.p.i.d. process. At some point in time, due to an anomaly, 
the distribution of the i.p.i.d. process deviates from $(f_1, \cdots, f_T)$. 
Our objective in this paper is to develop algorithms
to process $\{X_n\}$ in real time and detect changes in the distribution as quickly as possible, subject to a constraint 
on the rate of false alarms. 

In practice, the following parametric form of an i.p.i.d. process can be used. Let $\{X_n\}$ is an independent sequence of random variables with distribution in a parametric family with parameters $\{\theta_n\}$: 
\begin{equation}\label{eq:paraDistmod}
\begin{split}
X_n & \stackrel{\text{ind}}{\sim} p(\cdot; \theta_n), \; n \geq 1. \\
\end{split}
\end{equation}
The process $\{X_n\}$ is an i.p.i.d. process (a parametric i.p.i.d. process) if there is an integer $T > 0$ such that the parameter sequence $\{\theta_n\}$ is periodic with period $T$:
\begin{equation}\label{eq:periodPara}
\begin{split}
\theta_n &= \theta_{n+T},\; n \geq 1.
\end{split}
\end{equation}
Note that the statistical model in \eqref{eq:periodPara} has only $T$ parameters $\theta_1, \cdots, \theta_T$. The change detection problem in this case reduces
to detecting changes in these $T$ parameters. 
Given the parameters $\theta_1, \cdots, \theta_T$, we can use an algorithm to observe the process
$\{X_n\}$ sequentially over time $n$ and detect any changes in the values of any of the parameters. The baseline parameters in the problem, the period $T$ and the parameters within a period $\theta_1, \cdots, \theta_T$, can be learned from the training data. 
A special and important case of a parametric i.p.i.d. process 
is when we have a smooth function 
$\theta(t), t \in [0,1],
$
and 
\begin{equation}\label{eq:paraipidregression}
X_n \sim p\left(\cdot,\theta\left(\frac{n\%T}{T}\right) \right) \quad \; n \geq 1.
\end{equation}
where $n\%T$ represents $n$ modulo $T$. 
An example is a regression set up:
\begin{equation}\label{eq:paraipidregression_2}
X_n = \theta\left(\frac{n\%T}{T}\right) + Z_n \quad \; n \geq 1,
\end{equation}
where $\{Z_n\}$ is a zero-mean i.i.d. sequence. The change detection problem for the regression setup is 
the problem of detecting changes in the regression function $\theta(t)$.

Note that the sequence model in \eqref{eq:periodPara} is different from the sequence model studied in \cite{Johnstone2015Book} and \cite{tsybakov2009introduction}. In the model studied in \cite{Johnstone2015Book} and \cite{tsybakov2009introduction}, the random variables $\{X_n-\theta_n\}$ are modeled as Gaussian random variables and the parameters $\{\theta_n\}$ are not periodic. Furthermore, the problem there is of simultaneous estimation of all the different parameters $\{\theta_n\}$ given all the observations $\{X_n\}$. That is, the problem is not sequential in nature. It is also not a change point problem.

An i.p.i.d. process is a cyclostationary process  \cite{gardner2006cyclostationarity}. Although a cyclostationary process can also be used to model statistically periodic behavior, the i.p.i.d. definition captures sample level behavior and the independence assumption across time allows for the development of a strong change point detection theory.



\subsection{Change Point Model}
To define a change point model, consider another periodic sequence of densities $\{g_n\}$ such that 
$$
g_{n+T} = g_{n}, \quad \forall n \geq 1. 
$$
Thus, we essentially have $T$ distinct set of densities $(g_1, \cdots, g_T)$. We assume that at some point in time $\nu$, 
called the change point in the following, the law of the i.p.i.d. process is governed not by the densities $(f_1, \cdots, f_T)$, 
but by the new set of densities $(g_1, \cdots, g_T)$. 
 These densities need not be all different from the set of densities $(f_1, \cdots, f_T)$, 
but we assume that there exists at least an $i$ such that they are different:
\begin{equation}\label{eq:diffpdfassum}
g_i \neq f_i, \quad \text{for some } i = 1, 2, \cdots, T. 
\end{equation}
The change point model is as follows. At a time point $\nu$, the distribution of the random variable changes from $\{f_n\}$ to $\{g_n\}$:
\begin{equation}\label{eq:changepointmodel}
X_n \sim 
\begin{cases}
f_n, &\quad \forall n < \nu, \\
g_n &\quad \forall n \geq \nu.
\end{cases}
\end{equation}
We emphasize that the densities $\{f_n\}$ and $\{g_n\}$ are periodic. This model is equivalent to saying that we have two i.p.i.d. processes, 
one governed by the densities $(f_1, \cdots, f_T)$ and another governed by the densities $(g_1, \cdots, g_T)$, and at the change point $\nu$, 
the observation process switches from one i.p.i.d. process to another. 


\subsection{Problem Formulation}
We want to detect the change described in \eqref{eq:changepointmodel} as quickly as possible, subject to a constraint on the rate of false alarms. 
We are looking for a stopping time $\tau$ for the process to minimize a metric on the delay $\tau - \nu$ and to avoid the event of false alarm $\{\tau < \nu\}$. Specifically, 
we are interested in the popular false alarm and delay metrics of Pollak \cite{poll-astat-1985} and Lorden~\cite{lord-amstat-1971}. Let $\Prob_\nu$ denote the
probability law of the process $\{X_n\}$ when the change occurs at time $\nu$ and let $\Expect_\nu$ denote the corresponding expectation. When there is no change, 
we use the notation $\Expect_\infty$. The quickest change detection problem formulation of Pollak \cite{poll-astat-1985} is defined as 
\begin{equation}\label{eq:Pollak}
\begin{split}
\min_\tau &\;\;\; \sup_{\nu\geq 1} \; \Expect_\nu [\tau - \nu | \tau \geq \nu] \\
\text{subj. to}& \;\;\;\; \Expect_\infty[\tau] \geq \beta,
\end{split}
\end{equation}
where $\beta$ is a given constraint on the mean time to false alarm. Thus, the objective is to find a stopping time $\tau$ 
that minimizes the worst case conditional average detection delay subject to a constraint on the mean time to false alarm. 
A popular alternative is the worst-worst case delay metric of Lorden \cite{lord-amstat-1971}: 
\begin{equation}\label{eq:Lorden}
\begin{split}
\min_\tau &\;\;\;\; \sup_{\nu\geq 1} \; \text{ess} \sup \Expect_\nu [\tau - \nu | X_1, \cdots, X_{\nu-1}]\\
\text{subj. to}& \;\;\; \;\Expect_\infty[\tau] \geq \beta,
\end{split}
\end{equation}
where $\text{ess} \sup$ is used to denote the supremum of the random variable $\Expect_\nu [\tau - \nu | X_1, \cdots, X_{\nu-1}]$ outside a set of measure zero. 
Further motivation and comparison of these and other problem formulations for change point detection can be found in the literature \cite{veer-bane-elsevierbook-2013}, \cite{tart-niki-bass-2014}, \cite{poor-hadj-qcd-book-2009}, \cite{lai-ieeetit-1998}.

In Section~\ref{sec:SingleSeqOpt}, we develop algorithms and optimality theory for detecting changes in 
 i.p.i.d. processes. In Section~\ref{sec:UnknownPost}, we extend the results to the case when the post-change i.p.i.d. law $(g_1, \cdots, g_T)$ is unknown. 
In Section~\ref{sec:DistDetect}, we study the distributed case where there are multiple parallel streams of i.p.i.d. processes and the change can occur in any one of them. 

\section{Change Detection in a Single Sequence with Known Post-Change Law}\label{sec:SingleSeqOpt}
We now propose a CUSUM-type scheme to detect the above change. This algorithm belongs to the class of generalized CUSUM schemes discussed 
in the literature \cite{lai-ieeetit-1998}. We compute the sequence of statistics 
\begin{equation}\label{eq:PeriodicCUSUM}
W_{n+1} = \max_{1 \leq k \leq n+1} \sum_{i=k}^{n+1} \log \frac{g_{i}(X_{i})}{f_{i}(X_{i})}
\end{equation}
and raise an alarm as soon as the statistic is above a threshold $A$:
\begin{equation}\label{eq:PeriodicCUSUMstop}
\tau_c = \inf \{n \geq 1: W_n > A\}.
\end{equation}
We show below that this scheme is asymptotically optimal for the minimax problem formulations in \eqref{eq:Pollak} and \eqref{eq:Lorden}. We first show 
that the statistic $W_n$ can be computed recursively and using finite memory. 
%
%
\medskip
\begin{lemma}\label{thm:recurs}
The statistic sequence $\{W_n\}$ can be recursively computed as 
\begin{equation}\label{eq:PeriodicCUSUMrecur}
W_{n+1} = W_n^{+} + \log \frac{g_{n+1}(X_{n+1})}{f_{n+1}(X_{n+1})},
\end{equation}
where $(x)^+ = \max\{x, 0\}$. 
Further, since the set of pre- and post-change densities $(f_1, \cdots, f_T)$ and $(g_1, \cdots, g_T)$ are finite, 
the recursion \eqref{eq:PeriodicCUSUMrecur} can be computed using finite memory needed to store these $2T$ densities, the current observation, 
and the past statistic. 
\end{lemma}
\begin{proof}
The proof is provided in the appendix. 
\end{proof}
\medskip
In the rest of the paper, we refer to \eqref{eq:PeriodicCUSUMrecur} to as the Periodic-CUSUM algorithm. Towards proving the optimality of the Periodic-CUSUM scheme, we obtain a universal lower bound on the performance of any stopping time 
for detecting changes in i.p.i.d. processes. Define 
\begin{equation}\label{eq:KLnumber}
I = \frac{1}{T}\sum_{i=1}^T D(g_i \; \| \; f_i),
\end{equation}
where $D(g_i \; \| \; f_i)$ is the Kullback-Leibler divergence between the densities $g_i$ and $f_i$. We assume 
that
$$
D(g_i \; \| \; f_i) < \infty, \quad \forall i=1, 2, \cdots, T,
$$
and
$$
0 < D(g_i \; \| \; f_i), \quad \text{ for some } i=1, 2, \cdots, T.
$$

\begin{theorem}\label{thm:LB}
Let the information number $I$ as defined in \eqref{eq:KLnumber} satisfy $0 < I < \infty$. Then, for any stopping time $\tau$ satisfying the false alarm constraint $\Expect_\infty[\tau] \geq \beta$, we have as $\beta \to \infty$
\begin{equation}\label{eq:LB}
\begin{split}
\sup_{\nu\geq 1} \; \text{ess} &\sup \Expect_{\nu} [\tau - \nu | X_1, \cdots, X_{\nu-1}] \\
& \geq  \sup_{\nu\geq 1} \; \Expect_{\nu} [\tau - \nu | \tau \geq \nu]  \; \geq \; \frac{\log \beta}{I} (1 + o(1)),
\end{split}
\end{equation}
where an $o(1)$ term is one that goes to zero in the limit as $\beta \to \infty$. 
\end{theorem}
\begin{proof}
The proof is provided in the appendix. 
\end{proof}

We now show that the Periodic-CUSUM scheme \eqref{eq:PeriodicCUSUM}--\eqref{eq:PeriodicCUSUMrecur} is asymptotically optimal for 
both the formulations \eqref{eq:Pollak} and \eqref{eq:Lorden}. 
\begin{theorem}\label{thm:UB}
Let the information number $I$ as defined in \eqref{eq:KLnumber} satisfy $0 < I < \infty$. Then, the  Periodic-CUSUM stopping time $\tau_c$ \eqref{eq:PeriodicCUSUM}--\eqref{eq:PeriodicCUSUMrecur} with $A=\log \beta$ satisfies the false alarm constraint, i.e., 
$$\Expect_\infty[\tau_c] \geq \beta,
$$ 
and as $\beta \to \infty$,
\begin{equation}\label{eq:UB}
\begin{split}
\sup_{\nu\geq 1} \; & \Expect_\nu [\tau_c - \nu | \tau_c \geq \nu]  \\
&\leq \sup_{\nu\geq 1} \; \text{ess} \sup \Expect_\nu [\tau_c - \nu | X_1, \cdots, X_{\nu-1}]  \\
&\leq  \frac{A}{I} (1 + o(1)) = \frac{\log \beta}{I} (1 + o(1)).
\end{split}
\end{equation}
\end{theorem}
\begin{proof}
The proof is provided in the appendix. 
\end{proof}


\section{Change Detection With Unknown Post-Change I.P.I.D. Law}\label{sec:UnknownPost}
In the previous section, we assumed that the post-change law $(g_1, \cdots, g_T)$ is known to the decision maker. This information 
was used to design the Periodic-CUSUM algorithm \eqref{eq:PeriodicCUSUMrecur}. In practice, this information may not be available. We 
now show that if the post-change law belongs to a finite set of $M$ possible distributions,
$(g_1^{(1)}, \cdots, g_T^{(1)}), \cdots, (g_1^{(M)}, \cdots, g_T^{(M)})$, then an asymptotically optimal test can be designed. 

For $\ell \in \{1, \cdots, M\}$, define the statistic
\begin{equation}\label{eq:PeriodicCUSUMrecurell}
W_{n+1}^{(\ell)} = \left(W_n^{(\ell)}\right)^{+} + \log \frac{g_{n+1}^{(\ell)}(X_{n+1})}{f_{n+1}(X_{n+1})},
\end{equation}
and the stopping rule
\begin{equation}
\tau_{c\ell} = \inf \left\{n \geq 1: W_{n}^{(\ell)}  \geq \log (\beta M)\right\},
\end{equation}
which is the Periodic-CUSUM stopping rule for the $\ell$th post-change law $(g_1^{(\ell)}, \cdots, g_T^{(\ell)})$.
Now, define
\begin{equation}
\tau_{cm} = \inf \left\{n \geq 1: \max_{1 \leq \ell \leq M} W_{n}^{(\ell)}  \geq \log (\beta M)\right\}.
\end{equation}
Then, note that
\begin{equation}\label{eq:UBmax}
\tau_{cm} = \min_{\ell =1, \cdots, M} \tau_{c\ell}. 
\end{equation}
The stopping rule $\tau_{cm}$ is the stopping rule under which we stop the first time any of the $\ell$ Periodic-CUSUMs is above the threshold $\log(\beta M)$.

We now show that this stopping rule is optimal for both Lorden's and Pollak's criteria. Towards this end, we define a
Shiryaev-Roberts-type statistic
\begin{equation}
R_n = \sum_{\ell=1}^M \sum_{k=1}^n \prod_{i=k}^n \frac{g_{i}^{(\ell)}(X_{i})}{f_{i}(X_{i})}
\end{equation}
and a Shiryaev-Roberts-type stopping rule
\begin{equation}
\tau_{sr} = \inf \left\{n \geq 1: R_n \geq \beta M\right\}.
\end{equation}
Note that 
\begin{equation}\label{eq:CUSUMleqSR}
\tau_{sr} \leq \tau_{cm}.
\end{equation}
We have the following theorem.
\begin{theorem}\label{thm:unknownpost}
The process $\{R_n - nM\}$ is a $\Prob_\infty$ martingale. If 
$
\Expect_\infty[R_{\tau_{sr}}] < \infty,
$
then
$$
\Expect_\infty[\tau_{cm}] \geq \Expect_\infty[\tau_{sr}] \geq \beta. 
$$
Further, if $(g_1^{(\ell)}, \cdots, g_T^{(\ell)})$ is the true post-change i.p.i.d. law and 
\begin{equation}\label{eq:KLnumberell}
I_\ell = \frac{1}{T}\sum_{i=1}^T D(g_i^{(\ell)} \; \| \; f_i),
\end{equation}
then as $\beta \to \infty$,
\begin{equation}\label{eq:UBMax}
\begin{split}
\sup_{\nu\geq 1} \; & \Expect_\nu [\tau_{cm} - \nu | \tau_{cm} \geq \nu]  \\
&\leq \sup_{\nu\geq 1} \; \text{ess} \sup \Expect_\nu [\tau_{cm} - \nu | X_1, \cdots, X_{\nu-1}]  \\
&\leq   \frac{\log \beta}{I_\ell} (1 + o(1)).
\end{split}
\end{equation}
\end{theorem}
\begin{proof}
The proof is provided in the appendix. 
\end{proof}

Given the lower bound in Theorem~\ref{thm:LB}, the stopping rule $\tau_{cm}$ is thus asymptotically optimal with respect to the criteria of Lorden and Pollak, 
uniformly over each possible post-change hypothesis $(g_1^{(\ell)}, \cdots, g_T^{(\ell)})$, $\ell = 1, \cdots, M$. 

\section{Change Detection in a Distributed I.P.I.D.  Setting}\label{sec:DistDetect}
In the previous sections, we assumed that there is a single sequence of random variable $\{X_n\}$. In many applications, as in the problem of event detection in NYC data,
the sensors are distributed. An event can occur near any one of the sensors. Thus, it is of interest to develop algorithms to detect event using multiple streams of
data. In this section, we obtain optimal algorithms for detecting changes when the observation process in each stream is an i.p.i.d. process. 

Let there be $M$ independent streams of data and let $\{X_{n,\ell}\}_{n=1}^\infty$ be the i.p.i.d. observation process of the $\ell$th stream with law 
$(f_1^{(\ell)}, \cdots, f_T^{(\ell)})$. At the change point $\nu$, the law of the i.p.i.d. process in one of the streams changes from 
$(f_1^{(\ell)}, \cdots, f_T^{(\ell)})$ to $(g_1^{(\ell)}, \cdots, g_T^{(\ell)})$. The objective is to detect this change in distribution. 

For $\ell \in \{1, \cdots, M\}$, define the statistic
\begin{equation}\label{eq:distPeriodicCUSUMrecurell}
D_{n+1}^{(\ell)} = \left(D_n^{(\ell)}\right)^{+} + \log \frac{g_{n+1}^{(\ell)}(X_{n+1, \ell})}{f_{n+1}^{(\ell)}(X_{n+1, \ell})},
\end{equation}
and 
\begin{equation}
\tau_{dm} = \inf \left\{n \geq 1: \max_{1 \leq \ell \leq M} D_{n}^{(\ell)}  \geq \log (\beta M)\right\}.
\end{equation}
Note that this rule is different from that discussed in \eqref{eq:PeriodicCUSUMrecurell} because
here the statistic $D_{n+1}^{(\ell)}$ utilizes a different stream of observations and a different pre-change distribution for each stream index $\ell$. 

Now, define the Shiryaev-Roberts statistic
\begin{equation}
S_n = \sum_{\ell=1}^M \sum_{k=1}^n \prod_{i=k}^n \frac{g_{i}^{(\ell)}(X_{i, \ell})}{f_{i}^{(\ell)}(X_{i, \ell})}
\end{equation}
and the Shiryaev-Roberts stopping rule
\begin{equation}
\tau_{srd} = \inf \left\{n \geq 1: S_n \geq \beta M\right\}.
\end{equation}
We have the following theorem.
\begin{theorem}\label{thm:dist}
If 
$
\Expect_\infty[S_{\tau_{srd}}] < \infty,
$
then
$$
\Expect_\infty[\tau_{dm}] \geq \beta. 
$$
Further, if the change occurs in the $\ell$th stream and 
\begin{equation}\label{eq:distKLnumberell}
I_\ell = \frac{1}{T}\sum_{i=1}^T D(g_i^{(\ell)} \; \| \; f_i^{(\ell)}),
\end{equation}
then as $\beta \to \infty$,
\begin{equation}\label{eq:distUBMax}
\begin{split}
&\sup_{\nu\geq 1} \;  \Expect_\nu [\tau_{dm} - \nu | \tau_{dm} \geq \nu]  \\
&\leq \sup_{\nu\geq 1} \; \text{ess} \sup \Expect_\nu [\tau_{dm} - \nu | X_{1,\ell}, \cdots, X_{\nu-1, \ell}; \; \ell=1,\cdots,M]  \\
&\leq  \frac{\log \beta}{I_\ell} (1 + o(1)).
\end{split}
\end{equation}
\end{theorem}
\begin{proof}
The proof is similar to the proof of Theorem~\ref{thm:unknownpost} and is skipped. 
\end{proof}
In view of the lower bound obtained in Theorem~\ref{thm:LB}, the above theorem shows that the stopping rule $\tau_{dm}$ is asymptotically optimal 
for each post-change scenario. 

\section{Learning a Parametric I.P.I.D. Model}\label{sec:Paramipid}
In practice, learning pre- and post-change laws $(f_1, \cdots, f_T)$ and $(g_1, \cdots, g_T)$ can be hard. Thus, data will be typically modeled using a parametric i.p.i.d. model. 
We recall the definition of a parametric i.p.i.d. process from \eqref{eq:paraDistmod}:
\begin{definition}
Let $\{X_n\}$ be an independent sequence of random variables with distribution in a parametric family with parameters $\{\theta_n\}$: 
\begin{equation}
\begin{split}
X_n & \stackrel{\text{ind}}{\sim} p(\cdot; \theta_n), \; n \geq 1. \\
\end{split}
\end{equation}
The process $\{X_n\}$ is called a parametric i.p.i.d. process if there is an integer $T > 0$ such that the parameter sequence $\{\theta_n\}$ is periodic with period $T$:
\begin{equation}
\begin{split}
\theta_n &= \theta_{n+T},\; n \geq 1.
\end{split}
\end{equation}
\end{definition}
If the statistical properties of the observation process 
repeats every day, and if the data is collected once per hour, then in the above model, the period $T$ would correspond to $T= 24$ hours in a day, and the variables $X_1, \cdots, X_{T}$ would correspond to the data collected in the first day. In many applications, 
the data is often collected more frequently, at the rate of many samples per second. In such applications, $T$ could be, for example, equal to $24 \times 60 \times 60 \times m$, where $m$ is the number of samples collected per second. In practice, it may be hard to learn a large number of parameters, and detect changes in them, especially if the post-change parameters are not known. The learning process can be made simpler by making additional assumptions on the model. We discuss one such simplification now. 

In order to control the complexity of the problem, we may
assume that the parameters are divided into batches and parameters in each batch are approximately constant. For example, a batch may correspond to data collected in an hour 
and the average count of objects may not change in an hour. Mathematically, we assume that in each cycle or period of length $T$, the vector of parameters
 $\{\theta_k\}_{k=1}^T$ is partitioned into $E$ batches or episodes. 
Specifically, for $N_0=0$ and positive integers $\{N_e\}_{e=1}^E$ we define 
$B_e = \{N_{e-1} + 1, \cdots, N_e\}$ such that 
$\{1, \cdots, T\} = \cup_{e=1}^E B_e, \quad B_e \cap B_f = \emptyset, \text{ for } e \neq f.$
For $e\in \{1, \cdots, E\}$, we define
$\theta_{B_e} = (\theta_{N_{e-1}+1}, \cdots, \theta_{N_e}).$
Thus, $\{\theta_k\}_{k=1}^T$ is partitioned as
\begin{equation}\label{eq:BatchModel}
\overbrace{\theta_1, \cdots, \theta_{N_1}}^{\theta_{B_1}},\overbrace{\theta_{N_1+1}, \cdots, \theta_{N_2}}^{\theta_{B_2}}, \cdots, \overbrace{\theta_{N_{E-1}+1}, \cdots, \theta_{N_E}}^{\theta_{B_E}}. 
\end{equation}
Note that we have $T = \sum_{e=1}^E |B_e|$. 

We further assume a step model for parameters. Under this assumption, the parameters remain constant within a batch resulting in the step-wise constant sequence model
\begin{equation}\label{eq:StepModel}
\overbrace{\theta^{(1)}, \cdots, \theta^{(1)}}^{\theta_{B_1}},\overbrace{\theta^{(2)}, \cdots, \theta^{(2)}}^{\theta_{B_2}}, \cdots, \overbrace{\theta^{(E)}, \cdots, \theta^{(E)}}^{\theta_{B_E}}. 
\end{equation}
That is $\theta^{(1)} = \theta_1=\cdots=\theta_{N_1}$, $\theta^{(2)} = \theta_{N_1+1}=\cdots= \theta_{N_2}$, and so on. Thus, if the batch sizes are large, there are only $E \ll T$ parameters to learn from the data. Also, we have $|B_e|$ samples for batch $e$. 
The objective is then to observe the process $\{X_n\}$ over time and detect any changes in the parameters $\theta^{(1)}, \cdots, \theta^{(E)}$. 

Note that for the regression model, the batch assumption above is equivalent to approximating the smooth function $\theta(t)$ by a step function. 

\section{Numerical Results}\label{sec:Numerical}
We now apply the stopping rule $\tau_c$ in \eqref{eq:PeriodicCUSUMstop} to the NYC Instagram count data. 
In Fig.~\ref{fig:AllTestConcatenatedCount}, 
we have plotted the evolution of the test statistic $W_n$ applied to the Instagram counts collected near nine CCTV cameras. Two of these CCTV cameras 
are near or on the path of the run (called on-path cameras) and the rest are outside the race path (called off-path cameras). Out of the two on-path cameras, 
one is near the north end of the run and another is at the south end of the run. 
To obtain Fig.~{fig:AllTestConcatenatedCount}, we arranged the Instagram counts near each CCTV camera in a concatenated fashion, with labeled segments separated via red vertical lines. Each day has $6598$ samples. To compute the statistic, we divided the data for each day into four batches, with the first three batches being of length $1500$. We modeled the data as a sequence of Poisson random variables. We used the count data from Sept. 10 (one of the non-event days) to learn the averages of these Poisson random variables for each of the four batches. We assumed that there is only one post-change parameter per batch with value equal to three times the value of the normal parameter for that batch. We then applied the test $\tau_c$ to all the four days of data. 
As observed in the top figure in Fig.~\ref{fig:AllTestConcatenatedCount}, the test statistic stays close to zero for all the cameras on the non-event days. On the event day (Sept. 24th), the test statistic grows for the on-path cameras indicating that a deviation from the baseline was detected. The most significant change was observed near the CCTV camera near the north end of the run. In the bottom figure of Fig.~\ref{fig:AllTestConcatenatedCount}, we have replotted the test statistic applied to the on-path camera near the south end of the run. 
\begin{figure}
\centering
\includegraphics[scale=0.5]{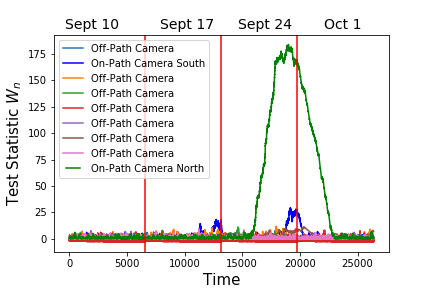}
\includegraphics[scale=0.5]{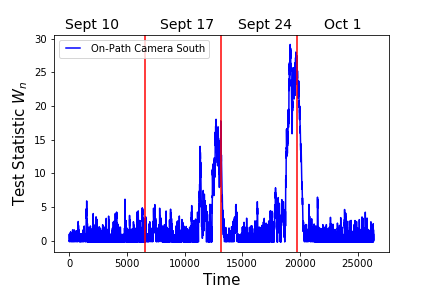}
\label{fig:Numerical}
\caption{Plots of test statistic $W_n$ applied to streams of Instragram counts collected near nine CCTV cameras. The on-path cameras are on the path of the 5K run. The statistic for on-path camera south is replotted in the bottom figure. }
\label{fig:AllTestConcatenatedCount}
\end{figure}

In Fig.~\ref{fig:DelayFAR}, we have plotted the delay $\Expect_1[\tau_c]$ versus log of mean time to false alarm $\log \Expect_\infty[\tau_c]$
for the periodic-CUSUM algorithm. The simulation plot was obtained for the following set of pre- and post-change parameters:
\begin{equation}\label{eq:delayMFApara}
\begin{split}
f_1 &= f_2 = \mathcal{N}(0,1),\\
g_1 &= \mathcal{N}(1,1),\\
g_2 &= \mathcal{N}(0.5,1). 
\end{split}
\end{equation}
To obtain each of the five points in the figure for simulations, the value of the threshold $A$ in \eqref{eq:PeriodicCUSUMstop} was set to 
values $3, 4,5,5.5$ and $6$ and both delay and false alarm estimates were obtained using $5000$ sample paths. The analysis plot was obtained by 
dividing the threshold by the average KL-divergence between the densities. 
In Fig.~\ref{fig:WnEvolution}, we have plotted a typical evolution of the algorithm applied 
to simulated data. This plot was obtained for the same set of parameters specified in \eqref{eq:delayMFApara}. 

\begin{figure}
\includegraphics[scale=0.5]{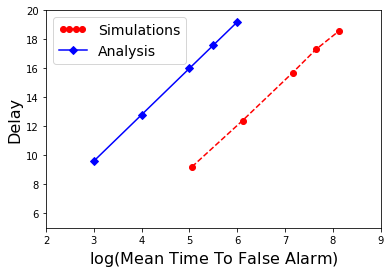}
\caption{Performance trade-off curves for the periodic-CUSUM algorithm. }
\label{fig:DelayFAR}
\end{figure}
\begin{figure}
\includegraphics[scale=0.5]{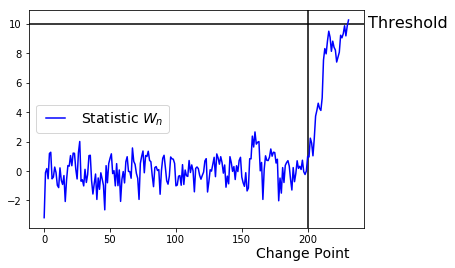}
\caption{Plots of test statistic $W_n$ from  \eqref{eq:PeriodicCUSUMrecur} applied to simulated data with threshold $A=10$.}
\label{fig:WnEvolution}
\end{figure}

\section{Conclusions}
We developed a minimax asymptotic theory for quickest detection of changes in i.p.i.d. models. We also studied the cases where the post-change i.p.i.d. law is unknown or 
where there are multiple streams of data. The algorithms developed were applied to count data extracted from multimodal data collected around a 5K run from NYC to detect
the 5K run. The theoretical results show that many of the results valid for i.i.d. data models are also true for i.p.i.d. models. An important question one 
can ask motivated from \cite{mous-astat-1986} is whether the periodic-CUSUM algorithm is exactly optimal for Lorden's formulation \cite{lord-amstat-1971}. The answer to this question, we believe, is negative. Our belief 
is based on the Bayesian analysis done in \cite{bane-ieeeit-2019} where we observed that single-threshold policies are not strictly optimal. In fact, in the Bayesian setting, the 
optimal algorithm has periodic thresholds. We conjecture that even for the minimax settings an algorithm with periodic thresholds will be strictly optimal.

\appendix

\begin{proof}[Proof of Lemma~\ref{thm:recurs}] For any sequence $\{Z_i\}$ of random variables, we can write
\begin{equation}
\begin{split}
\max_{1 \leq k \leq n+1} & \sum_{i=k}^{n+1} Z_i  = \max\left\{ \max_{1 \leq k \leq n} \sum_{i=k}^{n+1} Z_i, \; Z_{n+1}\right\}\\
=&\max\left\{ \max_{1 \leq k \leq n} \left(\sum_{i=k}^{n} Z_i + Z_{n+1}\right), \; Z_{n+1}\right\}\\
=&\max\left\{ \max_{1 \leq k \leq n} \left(\sum_{i=k}^{n} Z_i \right) + Z_{n+1}, \; Z_{n+1}\right\}\\
=&\max\left\{ \max_{1 \leq k \leq n} \left(\sum_{i=k}^{n} Z_i \right) , \; 0\right\} + Z_{n+1}.
\end{split}
\end{equation}
Substituting $Z_i = \log \frac{g_{i}(X_{i})}{f_{i}(X_{i})}$ into the above equation we get the desired recursion for $W_n$ in \eqref{eq:PeriodicCUSUM}: 
$$
W_{n+1} = W_n^{+} + \log \frac{g_{n+1}(X_{n+1})}{f_{n+1}(X_{n+1})}.
$$
Note that the increment term $\log \frac{g_{n+1}(X_{n+1})}{f_{n+1}(X_{n+1})}$ is only a function of the current observation $X_{n+1}$. Also, 
since the processes are i.p.i.d. with laws  $(f_1, \cdots, f_T)$ and $(g_1, \cdots, g_T)$, the likelihood ratio functions $ \frac{g_{n}(\cdot)}{f_{n}(\cdot)}$ 
are not all distinct, and there are only $T$ such functions $ \frac{g_{1}(\cdot)}{f_{1}(\cdot)}$ to $ \frac{g_{T}(\cdot)}{f_{T}(\cdot)}$. Thus, 
we need only a finite amount of memory to store the past statistic, current observation, and $2T$ densities to compute this statistic recursively. 
\end{proof}

\begin{proof}[Proof of Theorem~\ref{thm:LB}]
Let $Z_i = \log \frac{g_i(X_i)}{f_i(X_i)}$ be the log likelihood ratio at time $i$. We show that the sequence $\{Z_i\}$ satisfies the following statement: as $n \to \infty$, 
\begin{equation}\label{eq:thm1_1}
\begin{split}
\sup_{\nu \geq 1} \esssup \; \Prob_\nu & \left(  \max_{t \leq n} \sum_{i=\nu}^{\nu+t} Z_i \geq I(1+\delta)n \big| X_1, \cdots, X_{\nu-1}\right) \\
& \to 0, \quad \forall \delta > 0,
\end{split}
\end{equation}
where $I$ is as defined in \eqref{eq:KLnumber}. The lower bound then follows from Theorem 1 in \cite{lai-ieeetit-1998}. 
Towards proving \eqref{eq:thm1_1}, note that as $n \to \infty$
\begin{equation}\label{eq:thm1_2}
\begin{split}
\frac{1}{n}\sum_{i=\nu}^{\nu+n} Z_i \to \frac{1}{T}\Expect_1\left[ \sum_{i=1}^T \log \frac{g_i(X_i)}{f_i(X_i)}\right] = I, \quad \text{a.s.} \; \Prob_\nu, \; \; \forall \nu \geq 1. 
\end{split}
\end{equation}
The above display is true because of the i.p.i.d. nature of the observation processes. This implies that as $n \to \infty$
\begin{equation}\label{eq:thm1_3}
\begin{split}
 \max_{t \leq n} \frac{1}{n}\sum_{i=\nu}^{\nu+t} Z_i \to I, \quad \text{a.s.} \; \Prob_\nu, \; \; \forall \nu \geq 1.
\end{split}
\end{equation}
To show this, note that
\begin{equation}\label{eq:thm1_4}
\begin{split}
 \max_{t \leq n} \frac{1}{n}\sum_{i=\nu}^{\nu+t} Z_i  = \max \left\{ \max_{t \leq n-1} \frac{1}{n}\sum_{i=\nu}^{\nu+t} Z_i, \; \;  \frac{1}{n}\sum_{i=\nu}^{\nu+n} Z_i\right\}.
\end{split}
\end{equation}
For a fixed $\epsilon > 0$, because of \eqref{eq:thm1_2}, the LHS in \eqref{eq:thm1_3} is greater than $I(1-\epsilon)$ for $n$ large enough. Also, let the maximum on the LHS be achieved at a point $k_n$, 
then 
$$
 \max_{t \leq n} \frac{1}{n}\sum_{i=\nu}^{\nu+t} Z_i = \frac{1}{n}\sum_{i=\nu}^{\nu+k_n} Z_i  = \frac{k_n}{n} \frac{1}{k_n}\sum_{i=\nu}^{\nu+k_n} Z_i.
$$  
Now $k_n$ cannot be bounded because of the presence of $n$ in the denominator. This implies $k_n > i$, for any fixed $i$, and $k_n \to \infty$. Thus, $\frac{1}{k_n}\sum_{i=\nu}^{\nu+k_n} Z_i \to I$. Since $k_n/n \leq 1$, we have that the LHS in \eqref{eq:thm1_3} is less than $I(1+\epsilon)$, for $n$ large enough. This proves 
\eqref{eq:thm1_3}. 
To prove \eqref{eq:thm1_1}, note that due to the i.p.i.d. nature of the process
\begin{equation}\label{eq:thm1_5}
\begin{split}
\sup_{\nu \geq 1} \esssup \; \Prob_\nu & \left(  \max_{t \leq n} \sum_{i=\nu}^{\nu+t} Z_i \geq I(1+\delta)n \big| X_1, \cdots, X_{\nu-1}\right) \\
=\sup_{1 \leq \nu \leq T} \; \Prob_\nu & \left(  \frac{1}{n}\max_{t \leq n} \sum_{i=\nu}^{\nu+t} Z_i \geq I(1+\delta) \right) 
\end{split}
\end{equation}
The right hand side goes to zero because of \eqref{eq:thm1_3} and because the maximum on the right hand side in \eqref{eq:thm1_5} is over only finitely many terms. 
\end{proof}

\begin{proof}[Proof of Theorem~\ref{thm:UB}]
Again with $Z_i = \log \frac{g_i(X_i)}{f_i(X_i)}$, we show that the sequence $\{Z_i\}$ satisfies the following statement:
\begin{equation}\label{eq:thm2_1}
\begin{split}
\lim_{n \to \infty} \sup_{k \geq \nu \geq 1} \esssup \; \Prob_\nu & \left(  \frac{1}{n} \sum_{i=k}^{k+n} Z_i \leq I - \delta \big| X_1, \cdots, X_{\nu-1}\right) \\
& = 0, \quad \forall \delta > 0.
\end{split}
\end{equation}
The upper bound then follows from Theorem 4 in \cite{lai-ieeetit-1998}. 
To prove \eqref{eq:thm2_1}, note that due to the i.p.i.d nature of the process we have 
\begin{equation}\label{eq:thm2_2}
\begin{split}
&\sup_{k \geq \nu \geq 1} \esssup \; \Prob_\nu  \left(  \frac{1}{n} \sum_{i=k}^{k+n} Z_i \leq I - \delta \big| X_1, \cdots, X_{\nu-1}\right) \\
& = \sup_{\nu + T \geq k \geq \nu \geq 1} \Prob_\nu  \left(  \frac{1}{n} \sum_{i=k}^{k+n} Z_i \leq I - \delta \right) \\
& = \max_{1 \leq \nu \leq T} \max_{\nu \leq k \leq \nu+T} \Prob_\nu  \left(  \frac{1}{n} \sum_{i=k}^{k+n} Z_i \leq I - \delta \right) \\
\end{split}
\end{equation}
The right hand side of the above equation goes to zero for any $\delta$ because of \eqref{eq:thm1_2} and also because of the finite number of maximizations. 
The false alarm result follows directly from \cite{lai-ieeetit-1998} with $A=\log \beta$ because the likelihood ratios here also form a $\Prob_\infty$ martingale.
\end{proof}

\begin{proof}[Proof of Theorem~\ref{thm:unknownpost}]
That $\{R_n - nM\}$ is a $\Prob_\infty$ martingale can be proved by direct verification. For the false alarm proof, we assume that 
$\Expect_\infty[\tau_{sr}] < \infty$, otherwise the proof is trivial. Since $\Expect_\infty[R_{\tau_{sr}}] < \infty$, we have that 
$R_{\tau_{sr}} - \tau_{sr} M$ is integrable. Further, as $n \to \infty$, 
\begin{equation}
\begin{split}
\int_{\tau_{sr} > n} & |R_{\tau_{sr}} -  \tau_{sr} M| \; dP_\infty \leq \int_{\tau_{sr} > n}  R_{\tau_{sr}} + \tau_{sr} M \; dP_\infty \\
&\leq \beta M \Prob_\infty(\tau_{sr} > n) + \int_{\tau_{sr} > n}  \tau_{sr} M \; dP_\infty \; \to \; 0.
\end{split}
\end{equation}
Thus, by the optional sampling theorem \cite{wood-nonlin-ren-th-book-1982} and \eqref{eq:CUSUMleqSR} we have
$$
\Expect_\infty[\tau_{cm}] \geq \Expect_\infty[\tau_{sr}]  = \frac{\Expect_\infty[R_{\tau_{sr}}]}{M} \geq \frac{M\beta}{M} = \beta.
$$
The delay result is true because of \eqref{eq:UBmax}. 
\end{proof}

\balance

\bibliographystyle{ieeetr}



\bibliography{QCD_verSubmitted}

\end{document}